\documentclass[letterpaper, 10 pt, conference]{ieeeconf}  

\IEEEoverridecommandlockouts                              

\usepackage{graphics} 
\usepackage{epsfig} 
\usepackage{amsmath} 
\usepackage{amssymb}  
\usepackage{xcolor}
\usepackage{cite}

\DeclareMathOperator*{\im}{Im}
\DeclareMathOperator*{\diag}{diag}

\newtheorem{theorem}{Theorem}
\newtheorem{definition}{Definition}

\newtheorem{corollary}{Corollary}
\newtheorem{assumption}{Assumption}
\newtheorem{remark}{Remark}
\newtheorem{proposition}{Proposition}

\newcommand{\E}{\mathbb{E}}

\newcommand{\R}{\mathbb{R}}
\newcommand{\T}{^{\mathsf T}}

\title{\LARGE \bf
A Data-Driven Koopman-Behavioral Distance for Nonlinear Dynamical Systems
}

\author{Sarang Sutavani and Umesh Vaidya
\thanks{Umesh Vaidya acknowledges financial support from NSF CMMI Award 2531804. Sarang Sutavani is with the Dept. of Electrical Engr., Clemson University, Clemson SC, USA. Email: {\tt\small ssutava@clemson.edu}. 
Umesh Vaidya is with the Dept. of Mechanical Engr., Clemson University, Clemson SC, USA. Email: {\tt\small uvaidya@clemson.edu}.}%
}

\begin{document}

\maketitle
\thispagestyle{empty}
\pagestyle{empty}

\begin{abstract}

Comparing nonlinear dynamical systems directly from trajectory data remains challenging because finite-horizon trajectory representations are generally coordinate dependent. Motivated by recent behavioral subspace approaches for linear systems, this paper introduces a Koopman-behavioral distance for nonlinear stochastic systems using multi-rollout trajectory data. State measurements are lifted through a common observable dictionary, and dominant lifted behavioral subspaces are compared using Grassmannian distances without explicit Koopman operator identification. We establish coordinate invariance, connection with recently proposed linear behavioral distances under exact lifted closure, and robustness to approximate closure. Numerical examples demonstrate coordinate invariance and parameter discrimination.

\end{abstract}

\section{INTRODUCTION}

Comparing dynamical systems is a fundamental problem with implications for model validation, parameter identification, fault detection, digital twins, novelty detection, and retrieval of dynamical models from large databases \cite{mezic2004comparison,mehta2005stochastic,yu2009kullback,vaidya2005comparison}. For linear time-invariant (LTI) systems, an effective class of comparison tools is based on subspace geometry. Finite-horizon behaviors are represented by column spaces of Hankel matrices, and distances between systems are computed as Grassmannian distances between these subspaces. These behavioral subspace metrics are attractive because they are data-driven, model-free, and invariant under similarity transformations of the state \cite{willems2005fundamental,markovsky2021behavioral,padoan2025distances}.

A comparison method need not reconstruct a complete model. Many applications only require deciding whether two datasets are dynamically compatible, identifying the closest reference regime, or detecting departure from nominal behavior. This motivates a trajectory signature that is comparable across datasets and insensitive to coordinate artifacts.

Extending this paradigm to nonlinear systems is more difficult. Nonlinear dynamics are not preserved under linear coordinate changes, and finite-horizon trajectory sets generally do not admit finite-dimensional linear representations in the original state space. As a result, Hankel subspaces formed directly from trajectories can depend strongly on the chosen coordinates rather than the underlying dynamics \cite{padoan2025distances}.

Koopman operator theory provides a natural way to address this issue. By lifting the state through observables, nonlinear evolution can be represented as linear evolution in a higher-dimensional feature space \cite{koopman1931hamiltonian,mezic2005spectral,mezic2013analysis,huang2018feedback,vaidya2025koopman}. This idea underlies data-driven methods such as EDMD \cite{williams2015data}, Hankel-based Koopman methods and HAVOK \cite{brunton2017havok}, and neural Koopman representations \cite{lusch2018deep}.

Despite this progress, a direct data-driven behavioral distance for
comparing distinct nonlinear systems remains largely unexplored.
Existing operator-theoretic approaches compare transfer operators,
Koopman spectra, or trajectory probability laws
\cite{mehta2005stochastic,mezic2004comparison,yu2009kullback},
requiring operator or statistical estimation.
In contrast, the proposed framework compares the geometry of
finite-horizon lifted trajectory subspaces directly from data,
without estimating Koopman operators, spectra, or probability
distributions.

In this paper, we develop a Koopman-lifted behavioral metric for nonlinear stochastic systems of the form
\[
  x_{t+1}=f(x_t,\xi_t),
\]
where $\xi_t$ is an exogenous stochastic input. Given multi-rollout state trajectories, we lift the data through a common dictionary of observables $z_t=\psi(x_t)$ and construct finite-horizon lifted behaviors as column spaces of multi-rollout Hankel matrices. For finite stochastic data, the dominant left singular vector subspaces of these matrices are compared through principal-angle-based Grassmannian distances.

The lifting is used only to construct a trajectory signature for system comparison; it does not modify the dynamics or learn a correction or control law. A reference bank is needed only for applications such as parameter or regime matching, whereas the pairwise distance itself requires only the two datasets being compared.
The key idea is that lifting restores the structural setting in which subspace metrics are meaningful. Under exact sample-wise lifted closure, the proposed metric reduces to the classical linear finite-horizon behavioral distance. Under approximate closure, subspace perturbation theory \cite{davis1970rotation,wedin1972perturbation} yields a robustness bound governed by the ratios of residual magnitude to singular-value gap. The construction is invariant under nonlinear coordinate changes when the observables are transformed by pullback, while raw-state behavioral distances can remain coordinate dependent.

The main contributions of this paper are summarized as follows. We introduce a data-driven Koopman-behavioral distance for nonlinear stochastic dynamical systems that compares systems directly from trajectory data without identifying an explicit Koopman operator. 
We show that the proposed metric is invariant under nonlinear coordinate transformations when observables are transformed through the Koopman pullback. We establish theoretical connections with classical behavioral distance metrics by proving that the proposed framework reduces to the linear behavioral distance under stochastic Koopman closure and remains robust under approximate closure through subspace perturbation analysis.
Numerical examples demonstrate both coordinate invariance and parameter discrimination, illustrating that the proposed metric captures intrinsic dynamical differences. Taken together, these results establish a bridge between Koopman operator theory and behavioral system theory, extending geometric behavioral distances from linear systems to nonlinear stochastic dynamics.

\section{Problem Setting and Koopman Lifting}
\label{sec:setting}

Consider the discrete-time nonlinear stochastic system 
\begin{equation}
  x_{t+1}=f(x_t,\xi_t),
  \label{eq:nl_system_noise}
\end{equation}
where $x_t\in\R^n$ and $\{\xi_t\}$ is i.i.d. with common law $\nu$, exogenous, and independent of $x_0$. We assume access to full-state trajectories collected over multiple independent rollouts, possibly from different initial conditions. Let $\psi:\R^n\to\R^N$ be a dictionary of observables and define the lifted state
\begin{equation}
  z_t:=\psi(x_t)\in\R^N.
  \label{eq:lifted_state}
\end{equation}
For stochastic systems, the Koopman operator $\mathcal K$ acts on scalar observables $\varphi$ as \cite{Lasota}
\begin{equation}
  (\mathcal K\varphi)(x):=\E[\varphi(f(x,\xi))],
  \label{eq:stoch_koopman_def}
\end{equation}
where the expectation is taken with respect to {\color{black}the disturbance law $\nu$. The moments are assumed such that $\E\|\psi(x_t)\|^2<\infty$.}. Our construction does not require estimating $\mathcal K$ explicitly; the Koopman viewpoint serves to characterize when lifted subspaces provide informative system signatures. Let the dataset consist of $M$ rollouts
\[
  \mathcal D:=\Big\{\{x_t^{(j)}\}_{t=0}^{T_j}\Big\}_{j=1}^{M},
\]
with corresponding lifted trajectories $z_t^{(j)}=\psi(x_t^{(j)})$. Fix a horizon $L\in\mathbb N$. For each rollout $j$ and time $t$ satisfying $t+L-1\leq T_j$, define the stacked lifted segment
\begin{equation}
  \mathbf z_{t:t+L-1}^{(j)}:=
  \begin{bmatrix}
    z_t^{(j)}\\ z_{t+1}^{(j)}\\ \vdots\\ z_{t+L-1}^{(j)}
  \end{bmatrix}\in\R^{NL}.
  \label{eq:stacked_segment}
\end{equation}
Denote the corresponding generic random segment by
\[
  \mathbf Z_t:=
  \begin{bmatrix}z_t\T&z_{t+1}\T&\cdots&z_{t+L-1}\T\end{bmatrix}\T.
\]
Collecting all such segments across all rollouts yields the multi-rollout lifted Hankel matrix
\begin{equation}
\begin{split}
  H_L(\mathcal D):=\big[&\mathbf z_{0:L-1}^{(1)}\ \cdots\ \mathbf z_{T_1-L+1:T_1}^{(1)}\ \big|\ \cdots\ \big|\\
  &\mathbf z_{0:L-1}^{(M)}\ \cdots\ \mathbf z_{T_M-L+1:T_M}^{(M)}\big],
\end{split}
  \label{eq:multirollout_matrix}
\end{equation}
where $\tau:=\sum_{j=1}^{M}(T_j-L+2)$ is the total number of length-$L$ segments and $H_L(\mathcal D; \psi)\in\R^{NL\times\tau}$.
Each column of $H_L(\mathcal D; \psi)$ is one local length-$L$ trajectory. Pooling independent rollouts broadens coverage of initial conditions and noise realizations, although overlapping columns within a rollout remain dependent. Increasing $L$ adds temporal context but also raises the dimension $NL$ and reduces the number of available segments.

\begin{assumption}[Multi-rollout stochastic sampling]
\label{ass:stoch_sampling}
Each rollout is generated using an independent realization of the noise process and an initial condition independent across rollouts. In addition, either the initial conditions are sampled from an invariant distribution or a burn-in period is discarded so that the retained samples are approximately stationary. We assume $\E\|z_t\|^2<\infty$, and that the segment second-moment matrix $\mathcal M_L:=\E[\mathbf Z_t\mathbf Z_t\T]$ is nondegenerate on the subspace of interest.
\end{assumption}

In classical behavioral system theory for deterministic linear systems, the behavior is the set of all trajectories consistent with the system dynamics. For a finite horizon $L$, this behavior can be represented by the column space of a Hankel matrix constructed from sufficiently rich data.
In the stochastic setting considered here, the trajectories are random realizations with the noise process $\{\xi_t\}$. The multi-rollout Hankel matrix $H_L(\mathcal D; \psi)$ is a collection of samples of length-$L$ lifted trajectory segments. Its column space captures the span of the sampled lifted trajectories.
Under assumptions such as stationarity or ergodicity and sufficiently rich multi-rollout excitation, this empirical Hankel subspace approximates the population second-moment structure of lifted trajectories. The proposed metric therefore compares systems through the geometry of their lifted trajectory distributions rather than through individual trajectory realizations.

\begin{definition}[Empirical lifted finite-horizon behavior]
\label{def:empirical_behavior}
The empirical lifted finite-horizon behavior induced by $\mathcal D$ over horizon $L$ is
\begin{equation}
  \mathcal B_L(\mathcal D; \psi):=\im\!\big(H_L(\mathcal D; \psi)\big)\subset\R^{NL}.
  \label{eq:behavior_subspace}
\end{equation}
\end{definition}
When the observable map is fixed and common to a comparison, the dependence on $\psi$ is suppressed for brevity, so $H_L(\mathcal D)$ and $\mathcal B_L(\mathcal D)$ denote $H_L(\mathcal D;\psi)$ and $\mathcal B_L(\mathcal D;\psi)$, respectively; the same convention is used for $\mathcal B_{L,r}$.

At the population level, the corresponding $r$-dimensional signature is the leading eigenspace of the second-moment matrix $\mathcal M_L$. Since $H_L(\mathcal D;\psi)H_L(\mathcal D;\psi)\T/\tau$ is the empirical segment second-moment matrix, the dominant left singular vectors estimate these population directions under stationary or ergodic sampling.
Under stochastic forcing, the full empirical span can be enlarged by weak residual and finite-sample directions. If $H_L(\mathcal D;\psi)=\widehat U\Sigma_H\widehat V\T$ and $\widehat U_r$ contains the first $r$ left singular vectors, we therefore define the practical dominant behavior
\begin{equation}
  \mathcal B_{L,r}(\mathcal D;\psi):=\im(\widehat U_r).
  \label{eq:dominant_behavior_subspace}
\end{equation}
A prescribed numerical rank $r$ is held fixed across all datasets in a comparison. The full space $\mathcal B_L$ is retained for exact structural statements, while finite-sample comparisons use $\mathcal B_{L,r}$. The observable dictionary, feature normalization, horizon, and rank must be common within a comparison; otherwise changes in the representation can be mistaken for changes in the dynamics.

\section{Grassmannian Distances Between Lifted Behaviors}
\label{sec:grassmann}

Let $\mathcal B_1,\mathcal B_2\subset\R^d$ be subspaces, and let $U\in\R^{d\times k}$ and $V\in\R^{d\times\ell}$ have orthonormal columns spanning them. If $p:=\min(k,\ell)$ and
\begin{equation}
    \begin{aligned}
        U\T V=Q\Sigma R\T,
        \quad \Sigma=&\diag(\sigma_1,\ldots,\sigma_p),\\
        \quad 1\geq\sigma_1\geq\cdots\geq&\sigma_p\geq0,
    \end{aligned}
    \label{eq:svd_principal_angles}
\end{equation}
then the principal angles $\{\theta_i\}_{i=1}^{p}$ satisfy
\begin{equation}
  \cos\theta_i=\sigma_i,\qquad i=1,\ldots,p.
  \label{eq:principal_angles}
\end{equation}
Using principal angles, we consider the chordal, geodesic, and Procrustes distances
\begin{align}
    d_{\mathrm{ch}}(\mathcal B_1,\mathcal B_2)
    &=\left(|k-\ell|+\sum_{i=1}^{p}\sin^2\theta_i\right)^{1/2},
    \label{eq:chordal_distance}\\
    d_{\mathrm{geo}}(\mathcal B_1,\mathcal B_2)
    &=\left(\left(\frac{\pi}{2}\right)^2|k-\ell|+\sum_{i=1}^{p}\theta_i^2\right)^{1/2},
    \label{eq:geodesic_distance}\\
    d_{\mathrm{proc}}(\mathcal B_1,\mathcal B_2)
    &=\left(|k-\ell|+2\sum_{i=1}^{p}\sin^2(\theta_i/2)\right)^{1/2}.
    \label{eq:procrustes_distance}
\end{align}

\begin{definition}[Koopman-behavioral distance]
\label{def:koopman_behavioral_distance}
For two datasets $\mathcal D_1$ and $\mathcal D_2$, the Koopman-behavioral distance at horizon $L$ is
\begin{equation}
  d_\star^{\mathrm{NL}}(\mathcal D_1,\mathcal D_2;L)
  :=d_\star\!\left(\mathcal B_L(\mathcal D_1),\mathcal B_L(\mathcal D_2)\right),
  \label{eq:koopman_behavioral_distance}
\end{equation}
where $d_\star$ denotes one of \eqref{eq:chordal_distance}-\eqref{eq:procrustes_distance}.
For finite stochastic data, the practical distance is
\begin{equation}
  d_{\star,r}^{\mathrm{NL}}(\mathcal D_1,\mathcal D_2;L)
  :=d_\star\!\left(\mathcal B_{L,r}(\mathcal D_1),\mathcal B_{L,r}(\mathcal D_2)\right),
  \label{eq:truncated_koopman_behavioral_distance}
\end{equation}
where the same prescribed numerical rank $r$ is used for both datasets. Alternatively, a common deterministic energy threshold $\rho$ may define $r_i=r_\rho(\mathcal D_i)$ and $d_{\star,\rho}^{\mathrm{NL}}:=d_\star(\mathcal B_{L,r_1}(\mathcal D_1),\mathcal B_{L,r_2}(\mathcal D_2))$.
\end{definition}
Each nonlinear system is represented by a finite-dimensional behavioral space obtained from lifted trajectory data. The proposed distance therefore compares $r$ dominant behavioral directions rather than individual trajectories or identified models. The lifted Hankel matrix collects finite-horizon lifted trajectory segments, and its column space represents the corresponding family of realizable lifted behaviors. Grassmannian distances compare these families geometrically through their principal angles rather than through pointwise trajectory comparisons.
For equal-dimensional subspaces, all three distances vanish when the subspaces coincide and increase with the principal angles. The structural results apply to each choice; the simulations use Procrustes distance for the shear and chordal distance for the quadrotor.
Figure~\ref{fig:schema_behavior_comparison} summarizes the workflow from multi-rollout trajectories and common lifting to dominant lifted behavior subspaces and their Grassmannian comparison.

\begin{figure}
    \centering
    \includegraphics[width=0.99\linewidth]{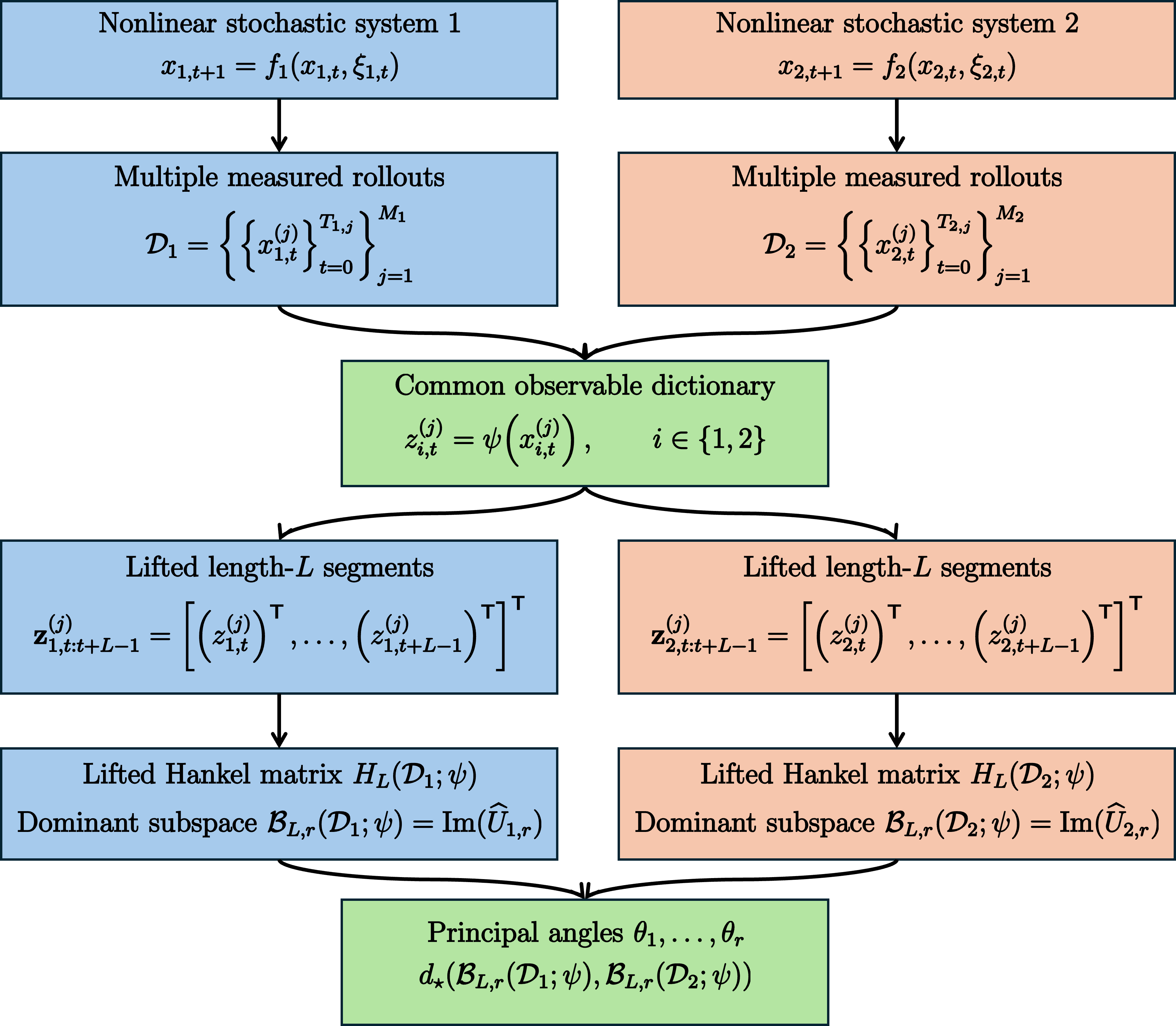}
    \caption{Koopman-behavioral comparison schematic: multi-rollout trajectories are lifted with common observables, assembled into Hankel matrices, reduced to dominant subspaces, and compared by Grassmannian distance.}
    \label{fig:schema_behavior_comparison}
\end{figure}

\subsection{Basic Invariance Properties}
\label{sec:invariance}

\begin{proposition}[Basic invariance properties]
\label{prop:invariance_multirun}
Fix a horizon $L\in\mathbb N$ and a lifting $\psi$. The full-space Koopman-behavioral distance depends only on the lifted behavior subspaces and is invariant under column reordering, invertible column reparameterizations, and common orthogonal changes of ambient coordinates.
\end{proposition}
\begin{proof}
Column reordering and invertible right multiplication preserve $\im(H_L)$. Common orthogonal left multiplication preserves all principal angles and hence the Grassmannian distances.
\end{proof}
For truncated SVD subspaces, column permutations and orthogonal right transformations preserve the dominant subspace, whereas a general invertible right transformation need not do so. Both the full and truncated distances remain invariant under a common orthogonal change of lifted coordinates.

\begin{remark}
Principal angles are generally not invariant under arbitrary invertible non-orthogonal changes of ambient coordinates. Such invariance can be recovered by whitening the data, or equivalently by using a data-induced inner product. We leave this extension to future work.
\end{remark}

\subsection{Invariance Under Diffeomorphic Coordinate Changes}
\label{subsec:stoch_invariance}

Consider the stochastic system
\begin{equation}
  x_{t+1}=f(x_t,\xi_t),
  \label{eq:stoch_sys_x}
\end{equation}
where $\{\xi_t\}_{t\geq0}$ is an exogenous stochastic process. Let $h:\R^n\to\R^n$ be a $C^1$ diffeomorphism and define the transformed coordinates $y=h(x)$. The induced dynamics are
\begin{equation}
  y_{t+1}=g(y_t,\xi_t):=h\!\left(f(h^{-1}(y_t),\xi_t)\right).
  \label{eq:stoch_sys_y}
\end{equation}
Fix a lifting $\psi:\R^n\to\R^N$ and define its pullback under $h$ by
\begin{equation}
  \widetilde\psi:=\psi\circ h^{-1},\qquad
  z_t:=\psi(x_t),\qquad
  \widetilde z_t:=\widetilde\psi(y_t).
  \label{eq:pullback_lift}
\end{equation}
Define the transformed stacked segment by $\widetilde{\mathbf Z}_t:=[\widetilde z_t\T\ \cdots\ \widetilde z_{t+L-1}\T]\T$. For paired rollouts, let $\mathcal D_x:=\{\{x_t^{(j)}\}_{t=0}^{T_j}\}_{j=1}^M$ and $\mathcal D_y:=\{\{y_t^{(j)}\}_{t=0}^{T_j}\}_{j=1}^M$, with $y_t^{(j)}=h(x_t^{(j)})$.

\begin{theorem}[Coordinate invariance under pullback lifting]
\label{thm:stoch_invariance_lifted}
For any horizon $L\in\mathbb N$, the lifted stacked segments satisfy $\widetilde{\mathbf Z}_t=\mathbf Z_t$ for all $t$. Consequently, for paired multi-rollout datasets with $y_t^{(j)}=h(x_t^{(j)})$,
\begin{equation}
  H_L(\mathcal D_x;\psi)=H_L(\mathcal D_y;\widetilde\psi),
  \label{eq:hankel_invariance}
\end{equation}
and hence
\begin{equation}
  \mathcal B_L(\mathcal D_x;\psi)=\mathcal B_L(\mathcal D_y;\widetilde\psi).
  \label{eq:empirical_subspace_invariance}
\end{equation}
The same conclusion holds for dominant subspaces obtained with a common deterministic rank-selection rule.
\end{theorem}

\begin{proof}
By definition, $\widetilde z_t=\widetilde\psi(y_t)=\psi(h^{-1}(y_t))=\psi(x_t)=z_t$. Stacking $L$ consecutive samples and applying the identity to every admissible segment proves \eqref{eq:hankel_invariance}-\eqref{eq:empirical_subspace_invariance}.
\end{proof}

\begin{corollary}[Population invariance under pullback lifting]
\label{cor:population_invariance}
If $\{x_t\}$ is stationary and $\E\|\psi(x_t)\|^2<\infty$, then the population segment second-moment matrices coincide:
\begin{equation}
  \widetilde\Sigma_L:=\E[\widetilde{\mathbf Z}_t\widetilde{\mathbf Z}_t\T]
  =\E[\mathbf Z_t\mathbf Z_t\T]=:\Sigma_L,
  \label{eq:covariance_invariance}
\end{equation}
and their dominant population subspaces are identical.
\end{corollary}

The theorem holds pathwise and does not rely on averaging over stochastic realizations. Independent datasets can still have nonzero empirical distance because their dominant subspaces are estimated from different samples, as examined in Example~1.

The invariance in Theorem~\ref{thm:stoch_invariance_lifted} relies on transforming the observables by pullback. If one compares trajectory subspaces directly in the raw state coordinates, this invariance is generally lost.

\begin{proposition}[Raw-state coordinate dependence]
\label{thm:stoch_noninvariance_raw}
Let $y_t=h(x_t)$ be a nonlinear diffeomorphic coordinate change. The $x$- and $y$-systems then describe the same underlying dynamics. Let $\mathcal B^{x}_{L,r}$ and $\mathcal B^{y}_{L,r}$ be the dominant finite-horizon subspaces constructed from their raw stacked trajectories. In general,
\[
  \mathcal B^{x}_{L,r}\neq\mathcal B^{y}_{L,r}.
\]
Consequently, $d_\star(\mathcal B^{x}_{L,r},\mathcal B^{y}_{L,r})$ need not vanish, and a raw-state behavioral distance can report a difference between two coordinate representations of the same dynamical system.
\end{proposition}

\begin{proof}
For nonlinear $h$, the stacked raw trajectory $\mathbf Y_t=[h(x_t)\T\ \cdots\ h(x_{t+L-1})\T]\T$ is generally not obtained from $\mathbf X_t$ through a fixed orthogonal linear map. Consequently, the corresponding raw-state second-moment matrices and their dominant subspaces need not be related by an orthogonal transformation, so their principal angles are not generally preserved.
\end{proof}

This contrast motivates the pullback construction, which removes coordinate distortion before the behavioral subspaces are compared.

\subsection{Reduction Under Exact Sample-Wise Lifted Closure}
\label{sec:reduction}

The exact reduction requires sample-wise closure; conditional-mean closure alone does not place noisy trajectory segments in a deterministic linear behavior.

\begin{assumption}[Exact sample-wise lifted closure]
\label{ass:koopman_closure_stochastic}
There exists $A\in\R^{N\times N}$ such that
\begin{equation}
  z_{t+1}=Az_t\quad\text{almost surely},\qquad z_t=\psi(x_t).
  \label{eq:stochastic_koopman_closure}
\end{equation}
\end{assumption}
Define
\begin{equation}
  \mathcal O_L(A):=\begin{bmatrix}I_N&A\T&\cdots&(A^{L-1})\T\end{bmatrix}\T
  \in\R^{NL\times N},
  \label{eq:finite_horizon_lifting_matrix}
\end{equation}
where $I_N$ is the identity matrix of size $N \times N$. $\mathcal O_L(A)$ maps each starting lifted state to its length-$L$ segment, so the behavior is determined by the admissible starting-state span.

If $\mathcal S\subseteq\R^N$ is the admissible lifted-state span, define
\begin{equation}
  \mathcal B_L^{\mathrm{lin}}(A;\mathcal S):=\mathcal O_L(A)\mathcal S.
  \label{eq:ideal_lti_behavior}
\end{equation}

\begin{assumption}[Horizon-$L$ richness of multi-rollout data]
\label{ass:richness_multirun}
Let $Z_{\mathcal D}$ collect the lifted states that start the columns of $H_L(\mathcal D)$. These states satisfy $\im(Z_{\mathcal D})=\mathcal S$.
\end{assumption}

\begin{theorem}[Reduction to linear behavior]
\label{thm:reduction_linear}
Under Assumptions~\ref{ass:koopman_closure_stochastic} and \ref{ass:richness_multirun},
\begin{equation}
  H_L(\mathcal D)=\mathcal O_L(A)Z_{\mathcal D},\qquad
  \mathcal B_L(\mathcal D)=\mathcal B_L^{\mathrm{lin}}(A;\mathcal S).
  \label{eq:behavior_equals_ideal}
\end{equation}
Consequently, for two datasets satisfying the assumptions,
\begin{equation}
  d_\star^{\mathrm{NL}}(\mathcal D_1,\mathcal D_2;L)
  =d_\star\!\left(\mathcal B_L^{\mathrm{lin}}(A_1;\mathcal S_1),
  \mathcal B_L^{\mathrm{lin}}(A_2;\mathcal S_2)\right).
  \label{eq:reduction_identity}
\end{equation}
\end{theorem}

\begin{proof}
Every Hankel column satisfies $\mathbf z_{t:t+L-1}=\mathcal O_L(A)z_t$, so $H_L(\mathcal D)=\mathcal O_L(A)Z_{\mathcal D}$. Taking images and applying the richness assumption gives \eqref{eq:behavior_equals_ideal}; substituting into \eqref{eq:koopman_behavioral_distance} gives \eqref{eq:reduction_identity}.
\end{proof}

Richness may be supplied collectively by many short rollouts. This result is not used to estimate a lifted linear model. Rather, it establishes that when $A$ exists, the data-driven behavior coincides with that of a linear system in the lifted coordinates.

\begin{remark}[Conditional-mean closure]
If only $\E[z_{t+1}\mid x_t]=Az_t$ holds, then $z_{t+1}=Az_t+\eta_t$ with $\E[\eta_t\mid x_t]=0$. Sampled Hankel columns generally do not belong to $\mathcal B_L^{\mathrm{lin}}(A;\mathcal S)$, so no exact equality of behavior spaces is claimed. This case is treated by the residual model below when the induced Hankel perturbation is small relative to the relevant singular-value gap.
\end{remark}

\subsection{Robustness Under Approximate Closure}
\label{sec:robustness}

We now quantify the deviation from the ideal lifted linear behavior when exact closure is unavailable.

\begin{assumption}[Linear recursion with residual]
\label{ass:approx_closure}
There exists $A\in\R^{N\times N}$ such that
\begin{equation}
  z_{t+1}=Az_t+\eta_t,
  \label{eq:approx_closure}
\end{equation}
where $\eta_t$ captures modeling error and/or stochastic forcing in lifted coordinates. Let $\bar H_L(\mathcal D)$ be generated by the ideal recursion using the same initial lifted states. Then
\[
  H_L(\mathcal D)=\bar H_L(\mathcal D)+E_L(\mathcal D).
\]
\end{assumption}

Iterating \eqref{eq:approx_closure} gives
\begin{equation}
  z_{t+k}=A^kz_t+\sum_{s=0}^{k-1}A^{k-1-s}\eta_{t+s},
  \qquad k=1,\ldots,L-1.
  \label{eq:residual_accumulation}
\end{equation}
Thus each column of $E_L(\mathcal D)$ stacks residuals accumulated over one segment. Increasing $L$ exposes more residual terms and raises the subspace dimension; $E_L(\mathcal D)$ may represent stochastic forcing, dictionary mismatch, or both.

\begin{theorem}[Subspace perturbation bound]
\label{thm:robustness}
Let $\bar{\mathcal B}_{L,r}(\mathcal D)$ and $\mathcal B_{L,r}(\mathcal D)$ be the dominant $r$-dimensional left singular subspaces of $\bar H_L(\mathcal D)$ and $H_L(\mathcal D)$, respectively. Define
\[
  \Theta\!\left(\mathcal B_{L,r}(\mathcal D),\bar{\mathcal B}_{L,r}(\mathcal D)\right):=\diag(\theta_1,\ldots,\theta_r),
\]
where $\theta_1,\ldots,\theta_r$ are their principal angles. Suppose
\begin{equation}
    \begin{aligned}
          \Delta(\mathcal D):=\sigma_r(\bar H_L(\mathcal D))-\sigma_{r+1}(\bar H_L(\mathcal D))>0,\\
          \|E_L(\mathcal D)\|_2<\Delta/2.
    \end{aligned}
    \label{eq:robustness_gap_condition}
\end{equation}
Then
\begin{equation}
  \|\sin\Theta(\mathcal B_{L,r}(\mathcal D),\bar{\mathcal B}_{L,r}(\mathcal D))\|_2
  \leq \frac{2\|E_L(\mathcal D)\|_2}{\Delta(\mathcal D)},
  \label{eq:sin_theta_robustness}
\end{equation}
and for each selected Grassmannian distance there is a rank-dependent constant $C_{\star,r}>0$ such that
\begin{equation}
  d_\star(\mathcal B_{L,r}(\mathcal D),\bar{\mathcal B}_{L,r}(\mathcal D))
  \leq C_{\star,r}\frac{\|E_L(\mathcal D)\|_2}{\Delta(\mathcal D)}.
  \label{eq:distance_robustness}
\end{equation}
\end{theorem}

\begin{proof}
By Weyl's singular-value perturbation inequality \cite{horn1991topics}, $|\sigma_j(H_L)-\sigma_j(\bar H_L)|\leq\|E_L\|_2$; hence Wedin's separation is at least $\Delta(\mathcal D)-\|E_L(\mathcal D)\|_2$. Wedin's sin-$\Theta$ theorem \cite[Thm.~4.1]{wedin1972perturbation} gives
\[
  \|\sin\Theta\|_2\leq\frac{\|E_L(\mathcal D)\|_2}{\Delta-\|E_L(\mathcal D)\|_2}
  \leq\frac{2\|E_L(\mathcal D)\|_2}{\Delta(\mathcal D)}.
\]
Each selected Grassmannian distance is bounded by a rank-dependent constant times the norm of the sine of the principal-angle vector.
\end{proof}

\begin{corollary}[Robustness of the pairwise distance]
\label{cor:pairwise_robustness}
For two datasets $\mathcal D_1$ and $\mathcal D_2$ satisfying Theorem~\ref{thm:robustness}, let $\Delta_i:=\Delta(\mathcal D_i)$ for $i=1,2$. Then
\begin{equation}
    \begin{aligned}
    &\left|d_\star(\mathcal B_{L,r}(\mathcal D_1),\mathcal B_{L,r}(\mathcal D_2))
    -d_\star(\bar{\mathcal B}_{L,r}(\mathcal D_1),\bar{\mathcal B}_{L,r}(\mathcal D_2))\right|\\
    &\quad\leq C_{\star,r}\left(
    \frac{\|E_L(\mathcal D_1)\|_2}{\Delta_1}+
    \frac{\|E_L(\mathcal D_2)\|_2}{\Delta_2}\right).
    \end{aligned}
    \label{eq:pairwise_robustness}
\end{equation}
\end{corollary}
\begin{proof}
With $B_i:=\mathcal B_{L,r}(\mathcal D_i)$ and $\bar B_i:=\bar{\mathcal B}_{L,r}(\mathcal D_i)$, the triangle inequality in both directions bounds the absolute difference by $d_\star(B_1,\bar B_1)+d_\star(B_2,\bar B_2)$. Applying Theorem~\ref{thm:robustness} to these two terms gives \eqref{eq:pairwise_robustness}.
\end{proof}

Thus the measured pairwise distance remains close to its ideal counterpart when both ratios between the residual norm and the singular-value gap are small. The robustness conclusion is conditional on the small relative error because zero-mean residuals may still produce a non-negligible perturbation in the empirical Hankel matrix.
Data amount, dictionary choice, and horizon affect the empirical perturbation and singular-value gap, motivating the studies below.

\section{Simulation Results}
\label{sec:simulations}

This section illustrates three properties of the proposed Koopman-behavioral metric: (i) invariance under nonlinear coordinate transformations when appropriate lifted observables are used, (ii) non-invariance of raw-state behavioral metrics under nonlinear coordinate changes, and (iii) sensitivity of the lifted metric to genuine changes in system parameters.
The examples also separate algebraic invariance from statistical robustness. Within each comparison, the observable map, feature preprocessing, horizon, and dominant-subspace rule are held common across the systems. Distances are interpreted within a fixed dictionary; absolute values from dictionaries with different dimensions or scalings are not treated as directly comparable.

\subsection{Example 1: Nonlinear Coordinate Shear}

We begin with the stable linear stochastic system in $\R^2$
\begin{equation}
  z_{t+1}=Az_t+\eta_t,
  \label{eq:z_linear_system}
\end{equation}
\[
\text{where}\quad
A=\begin{bmatrix}0.95&0\\0&0.55\end{bmatrix},\quad
  \eta_t\sim\mathcal N(0,\sigma^2I_2),\quad \sigma=0.07.
\]
Since the eigenvalues of $A$ lie inside the unit disk, the process is stable. We introduce the nonlinear diffeomorphic coordinate transformation
\begin{equation}
  x=h(z)=\begin{bmatrix}z_1\\z_2+cz_1^2\end{bmatrix},\qquad
  h^{-1}(x)=\begin{bmatrix}x_1\\x_2-cx_1^2\end{bmatrix}.
  \label{eq:shear_map}
\end{equation}
To compare behaviors across coordinate systems, we use the identity lifting $\psi_z(z)=z$ and the pullback lifting $\psi_x(x)=h^{-1}(x)$. By construction, $\psi_x(x_t)=z_t$ for all $t$. For comparison, raw-state observables use $\psi_x^{\mathrm{raw}}(x)=x$.
The distinct stable modes of $A$ and $\sigma=0.07$ provide persistent excitation, while varying $c$ changes only the coordinate representation and not the underlying trajectories in $z$. We use $r=6$: panels (a)-(b) sweep $c\in[0,3]$ with $M=50$, $T=350$, and $L\in\{5,10,15\}$; panel (c) plots the total sample count $MT$ for $M\in\{5,20,50,100\}$, $T\in\{60,100,180,350\}$, and the same $L$. Here, $M$ denotes the number of independent rollouts, $T$ is the number of time samples in each rollout, and $L$ is the trajectory-window or behavioral horizon used to construct the lifted data matrices.
In Fig.~\ref{fig:example1}, paired pullback matrices in (a) coincide pathwise, whereas raw distances grow with $c$. For independent data in (b), the pullback baseline is nearly independent of $c$; in (c), it decreases with data while raw-coordinate bias persists. For independently sampled datasets, the estimated behavioral subspaces approach the same population subspace as the amount of data increases, although their distance may remain nonzero at finite sample sizes.

\begin{figure*}
  \centering
  \includegraphics[width=0.93\textwidth]{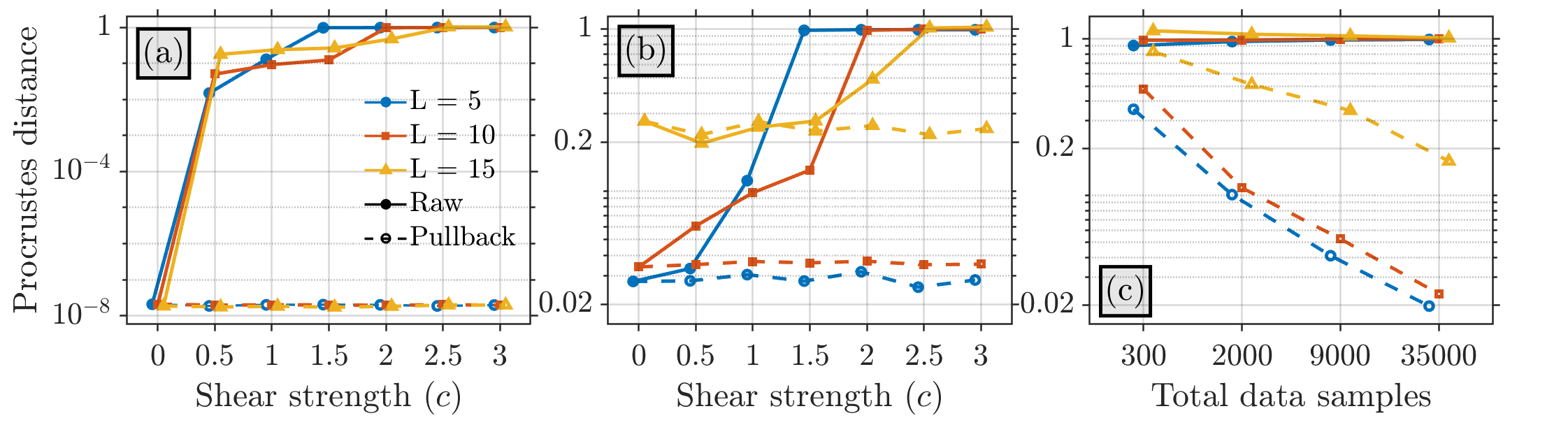}
  \caption{Nonlinear coordinate-shear experiment with dominant rank $r=6$ and horizons $L=5,10,15$. (a) Paired trajectories ($M=50$, $T=350$) share the same noise realizations; pullback lifting recovers identical finite-horizon behaviors up to numerical precision, while raw-state distances increase with shear strength $c$. (b) For independent trajectories with the same $M$ and $T$, pullback distances reflect sampling variability rather than coordinate distortion and remain approximately insensitive to $c$. (c) Across $M\in\{5,20,50,100\}$ and $T\in\{60,100,180,350\}$, increasing $MT$ reduces the independent-data pullback distance, whereas the raw-coordinate discrepancy remains.}\label{fig:example1}
\end{figure*}

\subsection{Example 2: Planar Quadrotor}

To demonstrate sensitivity to physical parameter changes, we consider a planar quadrotor with varying mass. We generate reference datasets for
\[
  m\in\{0.6,0.8,1.0,1.2,1.4\}\ \mathrm{kg},
\]
and query datasets at $m_q\in\{0.8,1.0,1.2\}$ kg. Using multiple interior queries avoids an endpoint-only or single-target conclusion.
The state is $x=[p_y,p_z,\phi,\dot p_y,\dot p_z,\dot\phi]\T$, where $p_y$ and $p_z$ denote lateral position and altitude, $\phi$ is the roll angle, and the remaining states are the corresponding velocities. The planar quadrotor dynamics are
\begin{equation}
\begin{aligned}
  \ddot p_y=\frac{-u_1\sin\phi+\eta_y}{m}, \;
  \ddot p_z=\frac{u_1\cos\phi+\eta_z}{m}-g, \;
  \ddot\phi=\frac{u_2}{J}.
\end{aligned}
\label{eq:quadrotor_dynamics}
\end{equation}
The system is discretized using a Runge-Kutta scheme. To generate informative trajectories, the quadrotor tracks a figure-8 reference under a fixed cascaded PD controller designed for the nominal mass $1.0$ kg. The resulting closed-loop system is treated as a noise-driven autonomous system, with $v_y$ and $v_z$ representing stochastic wind disturbances.
Because mass directly affects translational acceleration and the vertical motion shows most sensitivity to the parameter change based on our study, therefore we focus on the following observable dictionaries:
\[
\begin{aligned}
  \psi_{\mathrm{lin}}(x)
    &:=[p_z,\dot p_z]\T,\quad
  \psi_{\mathrm{elem}}(x)
    :=[p_z,\dot p_z,p_z^2,\dot p_z^2]\T,\\
  \psi_{\mathrm{full}}(x)
    &:=[p_z,\dot p_z,p_z^2,p_z\dot p_z,\dot p_z^2]\T.
\end{aligned}
\]
The dictionaries are fixed in advance and applied identically to every reference and query dataset. All panels use $r=12$ and $T=200$. Panel (a) uses $M=100$, $L=8$, and $m\in[0.6,1.4]$ kg. Panels (b),(d) use $L=8$ and $M\in\{25,50,100,200,500,800,1000,1500\}$; panel (c) uses $M=100$ and $L\in\{4,6,8,10,12\}$. The repeated-queries study the full nearest-bank comparison over independent trials.

Figure~\ref{fig:example2}(a) reports median chordal distances over $100$ independent trials. For each query mass, the minimum occurs at the matching reference-bank mass, showing that the metric detects a physical parameter change rather than a coordinate change. Panel (b) aggregates the three interior query masses and reports nearest-bank classification accuracy. Accuracy improves with $M$, and the quadratic dictionaries provide clearer separation in the limited-data regime because they represent nonlinear vertical dynamics not present in the linear dictionary. Panel (c), with $M=100$, isolates the horizon effect and indicates that performance is best near $L=6$-$8$. Short horizons contain less temporal information, whereas larger $L$ increases the lifted row dimension, leaves fewer overlapping length-$L$ segments per rollout, and exposes the estimate to more accumulated closure residual.

Let $\mathcal B_{L,r}^{\mathrm q}(m_q)$ and $\mathcal B_{L,r}^{\mathrm{ref}}(m)$ denote the query and reference dominant subspaces, and define
\[
\begin{aligned}
 d(m_q,m)
   &:=d_{\mathrm{ch}}\!\left(
      \mathcal B_{L,r}^{\mathrm q}(m_q),
      \mathcal B_{L,r}^{\mathrm{ref}}(m)\right),\\
 d_{\mathrm{true}}&:=d(m_q,m_q),\qquad
 d_{\mathrm{wrong}}:=\min_{m\neq m_q}d(m_q,m),\\
 \mu&:=d_{\mathrm{true}}-d_{\mathrm{wrong}}.
\end{aligned}
\]

Panel (d) uses $\mu=d_{\mathrm{true}}-\min_{m\neq m_q}d(m_q,m)$. Thus, $\mu<0$ means that the correct mass is closer than every incorrect candidate, and a more negative value indicates a larger identification margin. The increasingly negative margin with $M$ complements the accuracy curve by showing that the decision becomes not only more often correct but also better separated as the empirical subspaces stabilize. 
All candidates within a curve use the same features. Relevant nonlinear terms can improve separation, whereas unnecessary features may increase estimation variance.

\begin{figure*}
  \centering
  \includegraphics[width=0.98\textwidth]{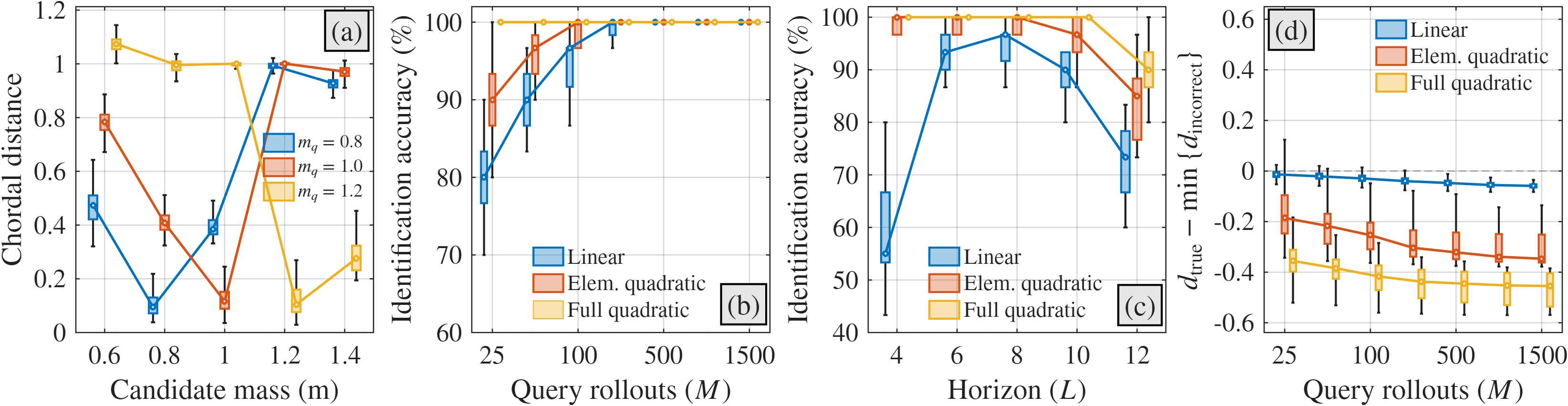}
  \caption{Planar-quadrotor mass discrimination using $r=12$ and rollout length $T=200$. (a) Median chordal distance over $100$ trials from each query mass to the five-mass reference bank ($M=100$, $L=8$); minima at the matching masses indicate correct identification. (b) Nearest-bank accuracy versus $M\in\{25,50,100,200,500,800,1000,1500\}$ at $L=8$, aggregated over query masses $0.8$, $1.0$, and $1.2$ kg. (c) Accuracy versus $L\in\{4,6,8,10,12\}$ at $M=100$, illustrating the tradeoff between temporal context and finite-sample estimation. (d) Margin $\mu=d_{\mathrm{true}}-d_{\mathrm{wrong}}$ versus the same $M$ values at $L=8$; negative values identify the correct mass, and larger negative magnitude indicates stronger separation.}\label{fig:example2}
\end{figure*}

\section{Conclusion and Future Directions}
\label{sec:discussion}

This paper developed a data-driven framework for comparing nonlinear dynamical systems from trajectory data by bringing together Koopman operator theory and behavioral system theory. Rather than comparing models, operators, or probability distributions, the proposed framework compares the geometry of finite-horizon behaviors, providing a coordinate-aware measure of similarity between nonlinear systems. This connection extends the scope of classical behavioral distances beyond linear systems and offers a new perspective for data-driven analysis of nonlinear dynamics.

The proposed framework also highlights the importance of selecting informative observable representations, since the quality of the comparison ultimately depends on how well the lifted coordinates capture the underlying system behavior. Future work will focus on controlled and time-varying systems, systematic construction of observable dictionaries, and applications to large-scale dynamical systems, model retrieval, fault detection, and digital twins. 






\bibliographystyle{ieeetr}
\bibliography{references}

\end{document}